\documentclass{article}
\usepackage{spconf,amsmath,graphicx,hyperref}

\usepackage{amsmath}
\usepackage{amsfonts}
\usepackage{algorithm}
\usepackage{algorithmic}
\usepackage{algorithm}
\usepackage{array}
\usepackage{textcomp}
\usepackage{stfloats}
\usepackage{url}
\usepackage{verbatim}
\usepackage{graphicx}
\usepackage{cite}
\usepackage{enumitem}
\usepackage{amsthm}
\usepackage[T1]{fontenc}
\usepackage{amssymb}
\usepackage{epsfig}
\usepackage{psfrag}
\usepackage{latexsym}
\usepackage{lettrine}
\usepackage{color}
\usepackage{multirow}
\usepackage{subfigure}
\usepackage{appendix}
\usepackage{enumitem}
\usepackage{bm}
\usepackage{array}
\usepackage{hyperref}					
\usepackage{balance}
\hypersetup{colorlinks,
	linkcolor=blue,%
	anchorcolor=blue,
    citecolor=blue}

\newtheorem{thm}{Theorem}
\newtheorem{assumption}{Assumption}

\newtheorem{lemma}{Lemma}
\newtheorem{corollary}{Corollary}

\newtheorem{proposition}{Proposition}

\title{discrete-periodic ambiguity function of Random \\
Communication Signals}
\name{Ying Zhang$^{\star}$, Fan Liu$^{\dagger}$, Yifeng Xiong$^{\star \dagger}$, Tao Liu$^{\star}$, and Shi Jin$^{\dagger}$}
\address{
$^{\star}$Southern University of Science and Technology, Shenzhen, China\\
$^{\dagger}$Southeast University, Nanjing, China\\
$^{\star \dagger}$Beijing University of Posts and Telecommunications, Beijing, China
}
\begin{document}
\ninept
\maketitle
\begin{abstract}
This paper investigates the ambiguity function (AF) of communication signals carrying random data payloads, which is a fundamental metric characterizing sensing capability in ISAC systems. We first develop a unified analytical framework to evaluate the AF of communication-centric ISAC signals constructed from arbitrary orthonormal bases and independent identically distributed (i.i.d.) constellation symbols. Subsequently, we derive the discrete periodic ambiguity function (DP-AF) and provide closed-form expressions for its expected integrated sidelobe level (EISL) and average sidelobe level. Notably, we prove that the normalized EISL is invariant across all constellations and modulation bases. Finally, the theoretical findings are validated through simulations.
\end{abstract}

\section{Introduction}
Integrated sensing and communication (ISAC) has emerged as a key enabling technology for 6G networks, offering support to
a wide range of next-generation applications\cite{ITU2023, challengesfor6G2023, liufan2022integrated}. A central challenge in ISAC is the design of dual-functional waveforms that simultaneously facilitate reliable communication and precise target sensing. ISAC waveform design strategies are typically grouped into sensing-centric, communication-centric, and joint design approaches \cite{7855671, Liufan2018optimalwavedesign,9540344}, where communication-centric designs are particularly attractive for future 6G ISAC deployments due to their low implementation complexity \cite{Wei2023}. However, since the communication signals must be intrinsically \textit{random} to deliver useful information, this gives rise to critical implementation challenges in communication-centric ISAC systems\cite{LushihangNetwork, LushihangTSP, Xiong2023}. Therefore, it is essential to investigate the sensing capabilities of communication-centric waveforms that transmit random data payloads.


Originally introduced by Woodward in 1953 \cite{woodward1953probability}, the ambiguity function (AF) characterizes the performance of deterministic radar waveforms in the delay-Doppler domain and serves as a key metric for assessing radar resolution and multi-target detection. While the AF of deterministic radar signals has been extensively studied \cite{7114336,7233208}, enabling precise analysis of their autocorrelation and time-frequency features, communication signals present unique challenges due to their inherent randomness. The randomness affects the statistical behavior of the AF and significantly influences target detection and parameter estimation performance. Accordingly, a rigorous analysis of the AF characteristics of random communication signals is necessary to support practical ISAC applications.


In our previous studies \cite{liu2024ofdmachieveslowestranging,11037613}, we analyzed the autocorrelation function (ACF) of random communication signals. The results indicated that orthogonal frequency-division multiplexing (OFDM) outperforms other waveforms in minimizing ACF sidelobe levels, thereby improving ranging accuracy. However, these studies were limited to characterizing the zero-Doppler AF characteristics, and a comprehensive analysis of the sensing performance of random communication signals in the range-Doppler domain remains unexplored. To address this gap, we further investigate the AF properties of random communication signals over the delay-Doppler domain.

This work analyzes the discrete periodic AF (DP-AF) of communication waveforms modulated by information-bearing symbols. We consider a monostatic ISAC setup, where the Tx emits a communication-centric ISAC signal by modulating random symbols over an orthonormal basis. The signal serves communication users while being reflected by targets to a sensing Rx collocated with the Tx. Thus, the sensing Rx fully knows the ISAC signal despite its randomness. 
The main contributions of this work are as follows. First, we propose a generic framework for analyzing the sensing performance of communication-centric ISAC waveforms constructed from orthonormal signaling bases and random data symbols. Specifically, we examine the DP-AF of random ISAC signals, using the average sidelobe level and the expected integrated sidelobe level (EISL) as performance metrics. Second, we derive general closed-form expressions for both the average sidelobe level and the EISL of the DP-AF, under arbitrary constellations and modulation formats. In particular, we prove that the normalized EISL remains constant across all constellations and modulation schemes.

{\emph{Notations}}: Matrices are denoted by bold uppercase (e.g., $\mathbf{U}$), vectors by bold lowercase (e.g., $\mathbf{x}$), and scalars by normal font (e.g., $N$). The $n$th entry of $\mathbf{s}$ and $(m,n)$-th entry of $\mathbf{A}$ are $s_n$ and $a_{m,n}$. $\otimes$, $\odot$, and $\operatorname{vec}(\cdot)$ denote the Kronecker product, Hadamard product, and vectorization. $\left(\cdot\right)^T$, $\left(\cdot\right)^H$, and $\left(\cdot\right)^*$ denote transpose, Hermitian transpose, and complex conjugate. The entry-wise square of $\mathbf{X}$ is $|\mathbf{X}|^2=\mathbf{X}\odot\mathbf{X}^*$. The $\ell_p$ norm is $\left\| \cdot\right\|_p$, and $\mathbb{E}(\cdot)$ denotes expectation. $\mathrm{Diag}(\mathbf{a})$ is the diagonal matrix with the entries of $\mathbf{a}$ on its main diagonal. $\mathbf{1}_{N,N}$ and $\mathbf{1}_{N}$ are the $N\times N$ all-one matrix and $N$-dimensional all-one vector. $\langle \cdot \rangle_N$ denotes the modulo $N$ operation. The Kronecker delta $\delta_{m,n}$ is
\begin{equation}
    \setlength{\abovedisplayskip}{5pt} 
    \setlength{\belowdisplayskip}{5pt} 
    \delta_{m,n} = \left\{ 
    \begin{aligned}
        & 0, \quad m \ne n; \\
        & 1, \quad m=n.
    \end{aligned}
    \right.
\end{equation}

\section{System Model and performance metrics}

\subsection{Signal Model}
Let us represent $N$ communication symbols to be transmitted as $\mathbf{s} = \left[s_1, s_2, \dots, s_N \right]^\mathrm{T} \in \mathbb{C}^{N}$. In typical communication systems, we shall modulate the $N$ symbols over an orthonormal basis on the time domain, which can be defined as a unitary matrix $\mathbf{U} = [\mathbf{u}_1, \mathbf{u}_2,...,\mathbf{u}_N] \in \mathbb{U}(N)$, where $\mathbb{U}(N)$ denotes the unitary group of degree $N$. Consequently, the discrete time-domian signal can be expressed as
\begin{equation} \label{eq:1-D signal}
    \setlength{\abovedisplayskip}{3pt} 
    \setlength{\belowdisplayskip}{3pt} 
    \mathbf{x} = \mathbf{U} \mathbf{s} = \sum_{n=1}^N s_n \mathbf{u}_n.
\end{equation}
The above generic model may represent most of the communication signaling schemes. 
To proceed, let us further impose some generic constraints on the considered constellation symbols.
\begin{assumption}[Unit Power and Rotational Symmetry]
    We focus on constellations with a unit power, zero mean, and zero pseudo-variance, defined as
    \begin{equation}
        \setlength{\abovedisplayskip}{4pt}
        \setlength{\belowdisplayskip}{4pt}
        \mathbb{E}( \left|s_n \right|^2) = 1, \quad\mathbb{E}(s_n) = 0,\quad\mathbb{E}(s_n^2) = 0,\quad \forall n.
    \end{equation}
\end{assumption}

\noindent The unit-power normalization in Assumption 1 ensures fair comparisons of sensing and communication performance under varying constellation formats. Moreover, most practical constellations, such as the PSK and QAM families, satisfy the zero-mean and zero pseudo-variance properties, whereas BPSK and 8-QAM do not satisfy the latter.

We now define the \textit{kurtosis} of a constellation, which serves as a key component in shaping the AF of communication signals, and is given by
\begin{equation}
    \mu_4 = \frac{\mathbb{E}\left\{|s_n-\mathbb{E}(s_n)|^4\right\}}{\mathbb{E}\left\{|s_n-\mathbb{E}(s_n)|^2\right\}^2}.
\end{equation}
Under unit power and zero mean conditions, kurtosis reduces to the fourth-order moment $\mu_4= \mathbb{E}(|s_n|^4)$. Using the power mean inequality, immediately yields
\begin{equation}
    \setlength{\abovedisplayskip}{4pt}
    \setlength{\belowdisplayskip}{4pt}
    \mathbb{E} \left( \left|s_n \right|^4 \right) \geq \left[ \mathbb{E}\left( \left|s_n\right|^2\right) \right]^2 = 1.
\end{equation}
Specifically, all PSK constellations exhibit $\mu_4 = 1$, while all QAM constellations have $1 \leq \mu_4 \leq 2$.

\subsection{Discrete Periodic Ambiguity Function (DP-AF)}
Let us consider the scenario where $P$ targets located at different ranges and with different moving speeds need to be sensed simultaneously, each of which is characterized by a complex amplitude $\beta_p$, a delay $\tau_p$, and a Doppler $\nu_p$. The echo signal at the sensing Rx is
\begin{equation}
    \setlength{\abovedisplayskip}{4pt}
    \setlength{\belowdisplayskip}{4pt}
    y(t) = \sum_{p=1}^P \beta_p x(t-\tau_p)e^{-j2 \pi \nu_p t} + z(t),
\end{equation}
where $x(t)$ is the radar signal in the continuous time domain and $z(t)$ is the added white Gaussian noise.

To extract the delay and Doppler parameters, a common practice is to matched-filter (MF) the echo signal $y(t)$ with the transmitted signal $x(t)$, yielding the following MF output in the continuous time domain
\begin{align}\label{MF output in the continuous time domain}
    \mathcal{MF}(\tau, \nu) 
    &\nonumber= \sum_{p=1}^P \int_{-\infty}^{\infty}  \beta_p x(t) x^* (t - (\tau-\tau_p)) e^{j 2\pi (\nu-\nu_p) t}\mathrm{d}t\\
    & + \int_{-\infty}^{\infty}z(t)  x^* (t - \tau) e^{j 2\pi \nu t}\mathrm{d}t. 
\end{align}    

Let us recall the Woodward's AF for the continuous-time signal $x(t)$, which is
\begin{equation} \label{AF output}
    \setlength{\abovedisplayskip}{4pt}
    \setlength{\belowdisplayskip}{4pt}
    \mathcal{X}(\tau,\nu) = \int_{-\infty}^{\infty} x(t) x^*(t-\tau) e^{j 2 \pi \nu t} dt.
\end{equation}
Combining \eqref{MF output in the continuous time domain} and \eqref{AF output}, one immediately observes that the MF output can be viewed as a linear combination of time- and frequency-shifted AFs plus noise. Toward that end, it is desired that the squared output $\left| \mathcal{MF}(\tau, \nu) \right|^2$ generates high peaks at $\tau=\tau_{p}$ and $\nu=\nu_p$, and small sidelobe levels elsewhere, which depend heavily on the overall geometry of the AF.

This paper mainly focuses on the case of the ISAC waveform with CP, where the length of the added CP is larger than the maximum delay of $Q$ targets. Accordingly, the DP-AF of the discrete-time signal $\mathbf{x}$ is
\begin{equation}
    \setlength{\abovedisplayskip}{4pt}
    \setlength{\belowdisplayskip}{4pt}
    \mathcal{X}(k,q) 
    = \sqrt{N} \mathbf{x}^H \mathbf{J}_k^H \mathrm{Diag}(\mathbf{f}_{q+1}^{*}) \mathbf{x}, 
\end{equation}
where $k,q = 0,1,2,..., N-1$, with $\mathbf{J}_k$ being the $k$th periodic shift matrix due to the addition of the CP, given as
\begin{equation}
    \setlength{\abovedisplayskip}{4pt}
    \setlength{\belowdisplayskip}{4pt}
    \mathbf{J}_k = \left[ \begin{matrix}
        \mathbf{0} & \mathbf{I}_k \\
        \mathbf{I}_{N-k} & \mathbf{0}
    \end{matrix}\right].
\end{equation}
The sidelobe level of $\mathcal{X}(k,q)$ is defined as
\begin{equation}\label{2D AF}
    \setlength{\abovedisplayskip}{4pt}
    \setlength{\belowdisplayskip}{4pt}
\begin{aligned}
    |\mathcal{X}(k,q)|^2 
    & = |\sqrt{N} \mathbf{x}^H \mathrm{Diag} (\mathbf{f}_{q+1}) \mathbf{J}_k \mathbf{x}|^2, \quad \forall k, q \ne 0.
\end{aligned}
\end{equation}
The mainlobe of DP-AF is $|\mathcal{X}(0,0)|^2 = |\mathbf{x}^H \mathbf{x}| = \|\mathbf{x}\|_2^4$. 
Due to the random nature of the signal, the DP-AF is a random function. Therefore, we define the average of the sidelobe level as a performance metric, which is expressed as
\begin{equation}
    \setlength{\abovedisplayskip}{4pt}
    \setlength{\belowdisplayskip}{4pt}
    \mathbb{E}(\left| \mathcal{X}(k,q) \right|^2) = \mathbb{E}(|\sqrt{N} \mathbf{x}^H \mathrm{Diag} (\mathbf{f}_{q+1}) \mathbf{J}_k \mathbf{x}|^2), \quad \forall k,q \ne 0,
\end{equation}
where the expectation is with respect to the random symbol vector $\mathbf{s}$. Accordingly, the average mainlobe is defined as \cite{liu2024ofdmachieveslowestranging} 
    \begin{equation}\label{the average mainlobe of AF}
        \setlength{\abovedisplayskip}{4pt}
        \setlength{\belowdisplayskip}{4pt}
        \mathbb{E}(\left| \mathcal{X}(0,0) \right|^2) = \mathbb{E}(\|\mathbf{x}\|_2^4) =N^2 + (\mu_4 - 1)N,
    \end{equation}
and the EISL is given by
\begin{equation}
    \setlength{\abovedisplayskip}{4pt}
    \setlength{\belowdisplayskip}{4pt}
\begin{aligned}
    \mathrm{EISL} 
    & =\sum_{k=0}^{N-1} \sum_{q=0}^{N-1} \mathbb{E}(|\mathcal{X}(k,q)|^2) - \mathbb{E}(|\mathcal{X}(0,0)|^2).
\end{aligned}
\end{equation}

\section{ Statistical Characterization of DP-AF }
\subsection{Average Sidelobe Level}
\begin{lemma}
The periodic shift matrix can be decomposed as
\begin{equation}
    \setlength{\abovedisplayskip}{5pt}
    \setlength{\belowdisplayskip}{5pt}
    \mathbf{J}_k = \sqrt{N} \mathbf{F}_N^H \mathrm{Diag}(\mathbf{f}_{k+1}) \mathbf{F}_N,
\end{equation}
where $\mathbf{F}_N$ is the normalized discrete Fourier transform (DFT) matrix of size $N$ and $\mathbf{f}_k$ is the $k$-th column of $\mathbf{F}_N$.
\end{lemma}
Using Lemma 1, \eqref{2D AF} can be rewritten as
\begin{equation}\label{2D AF -2}
    \setlength{\abovedisplayskip}{5pt}
    \setlength{\belowdisplayskip}{5pt}
    \ \left| \mathcal{X}(k,q) \right|^2 \
    = \left| \sqrt{N} (\mathbf{J}_{q} \mathbf{F}_N \mathbf{Us})^H \mathrm{Diag}(\mathbf{f}_{k+1}) (\mathbf{F}_N \mathbf{Us}) \right|^2.
\end{equation}

\begin{proposition}
The closed-form expression of average squared DP-AF is
\begin{align}\label{average sidelobe of AF}
    &\nonumber \mathbb{E}(\left| \mathcal{X}(k,q) \right|^2)
    =  \sum_{n=1}^N \sum_{m=1}^N \delta_{m,n} e^{-j\frac{2\pi}{N}k(n-m)} + (\mu_4 - 2) \times\\
    &\nonumber \sum_{n=1}^N \sum_{m=1}^N \mathbf{1}^T(\mathbf{v}_{\langle n-q\rangle_N} \odot \mathbf{v}_n^{*} \odot \mathbf{v}_{\langle m-q\rangle_N}^{*} \odot \mathbf{v}_m) e^{-j\frac{2\pi}{N}k(n-m)}  \\
    &+ \sum_{n=1}^N \sum_{m=1}^N   \delta_{q,0} e^{-j\frac{2\pi}{N}k(n-m)}  ,  
\end{align}
where $k,q = 0,1,2,..,N-1$ and $\mathbf{v}_n$ is the $n$-th column of $\mathbf{V} = \mathbf{U}^H \mathbf{F}_N^H$.

\begin{proof}
    We first rewrite \eqref{2D AF -2} as
\begin{equation}\label{the squared AF}
    \setlength{\abovedisplayskip}{4pt}
    \setlength{\belowdisplayskip}{4pt}
\begin{aligned}
    &\left| \mathcal{X}(k,q) \right|^2
    = \left| \sum_{n=1}^N \mathbf{s}^H \mathbf{v}_{\langle n-q\rangle_N} \mathbf{v}_n^H \mathbf{s}  e^{-j \frac{2\pi}{N} k (n-1)} \right|^2 \\
     &= \sum_{n=1}^N \sum_{m=1}^N \mathbf{v}_n^H \mathbf{s} \mathbf{s}^H \mathbf{v}_{\langle n-q\rangle_N}  \mathbf{v}_{\langle m-q\rangle_N}^H \mathbf{s} \mathbf{s}^H \mathbf{v}_m^H  e^{-j \frac{2\pi}{N} k (n-m)}.
\end{aligned}
\end{equation}
Define $\widetilde{\mathbf{s}} = \mathrm{vec}(\mathbf{s} \mathbf{s}^H) = \mathbf{s}^* \otimes \mathbf{s}$. Expanding $\mathbf{v}_n^H \mathbf{s} \mathbf{s}^H \mathbf{v}_{\langle n-q\rangle _N}$ yields
\begin{equation}\label{expanding 1}
    \setlength{\abovedisplayskip}{4pt}
    \setlength{\belowdisplayskip}{4pt}
    \mathbf{v}_n^H \mathbf{s} \mathbf{s}^H \mathbf{v}_{\langle n-q\rangle_N} 
    = (\mathbf{v}_{\langle n-q\rangle_N}^{T} \otimes \mathbf{v}_n^H) \widetilde{\mathbf{s}}. 
\end{equation}
Similarly, $\mathbf{v}_{\langle m-q\rangle _N}^H \mathbf{s} \mathbf{s}^H \mathbf{v}_m^H$ can be expanded as
\begin{equation}\label{expanding 2}
    \setlength{\abovedisplayskip}{4pt}
    \setlength{\belowdisplayskip}{4pt}
\begin{aligned}
    \mathbf{v}_{\langle m-q\rangle _N}^H \mathbf{s} \mathbf{s}^H \mathbf{v}_m^H 
    = \widetilde{\mathbf{s}}^H (\mathbf{v}_{\langle m-q\rangle _N}^{*} \otimes \mathbf{v}_m).   
\end{aligned}
\end{equation}
Plugging \eqref{expanding 1} and \eqref{expanding 2} into \eqref{the squared AF} immediately leads to
\begin{equation}
    \setlength{\abovedisplayskip}{4pt}
    \setlength{\belowdisplayskip}{4pt}
\begin{aligned}
    &\left| \mathcal{X}(k,q) \right|^2 \\
    &= 
    \sum_{n=1}^N \sum_{m=1}^N (\mathbf{v}_{\langle n-q \rangle_N}^{T} \otimes \mathbf{v}_n^H) \widetilde{\mathbf{s}} \widetilde{\mathbf{s}}^H (\mathbf{v}_{\langle m-q \rangle_N}^{*} \otimes \mathbf{v}_m)  e^{-j \frac{2\pi}{N} k (n-m)} .   
\end{aligned}
\end{equation}
Therefore, the average sidelobe is
\begin{equation}
    \setlength{\abovedisplayskip}{4pt}
    \setlength{\belowdisplayskip}{4pt}
\begin{aligned}
    &\mathbb{E}(\left| \mathcal{X}(k,q) \right|^2) = \\
    &\sum_{n=1}^N \sum_{m=1}^N (\mathbf{v}_{\langle n-q \rangle_N}^{T} \otimes \mathbf{v}_n^H) \mathbf{S} (\mathbf{v}_{\langle m-q \rangle_N}^{*} \otimes \mathbf{v}_m)  e^{-j \frac{2\pi}{N} k (n-m)}.
\end{aligned}
\end{equation}
where $\mathbf{S} = \mathbb{E}(\widetilde{\mathbf{s}} \widetilde{\mathbf{s}}^H)$, and can be decomposed as \cite{liu2024ofdmachieveslowestranging}
    \begin{equation} \label{S}
        \setlength{\abovedisplayskip}{4pt}
        \setlength{\belowdisplayskip}{4pt}
        \mathbf{S} = \mathbf{I}_{N^2} + \mathbf{S}_1 + \mathbf{S}_2 ,
    \end{equation}
where 
    \begin{align}
        \mathbf{S}_1 &= \mathrm{Diag} ( [\mu_4 - 2, \mathbf{0}_N^T, \mu_4 - 2, \mathbf{0}_N^T,..., \mu_4 - 2 ]^T ), \\
        \mathbf{S}_2 &= \left[\mathbf{g}, \mathbf{0}_{N^2 \times N},\mathbf{g},..., \mathbf{g}, \mathbf{0}_{N^2 \times N}, \mathbf{g} \right],
    \end{align}    
with $\mathbf{0}_{N^2\times N}$ being the all-zero matrix of size $N ^2 \times N$, and 
    \begin{equation}
        \setlength{\abovedisplayskip}{4pt}
        \setlength{\belowdisplayskip}{4pt}
        \mathbf{g} = \left[ 1,\mathbf{0}_N^T,1,...,1,\mathbf{0}_N^T,1 \right]^T.
    \end{equation}
Due to the fact that $\mathbf{v}_n^H \mathbf{v}_m = \delta_{m,n}$, we have
\begin{align} 
    &(\mathbf{v}_{\langle n-q \rangle_N}^{T} \otimes \mathbf{v}_n^H) \mathbf{I}_{N^2} (\mathbf{v}_{\langle m-q \rangle_N}^{*} \otimes \mathbf{v}_m) 
    = \delta_{m,n},  \label{term1} \\
    \nonumber&(\mathbf{v}_{\langle n-q \rangle _N}^{T} \otimes \mathbf{v}_n^H) \mathbf{S}_1 (\mathbf{v}_{\langle m-q \rangle _N}^{*} \otimes \mathbf{v}_m) \\
    &= (\mu_4 - 2)  \mathbf{1}^T (\mathbf{v}_{\langle n-q \rangle _N} \odot \mathbf{v}_n^{*} \odot \mathbf{v}_{\langle m-q \rangle _N}^{*} \odot \mathbf{v}_m) ,\label{term2} \\   
    \nonumber&(\mathbf{v}_{\langle n-q \rangle _N}^{T} \otimes \mathbf{v}_n^H) \mathbf{S}_2 (\mathbf{v}_{\langle m-q \rangle _N}^{*} \otimes \mathbf{v}_m) \\
    &= \mathbf{1}^T (\mathbf{v}_{\langle n-q \rangle _N} \odot \mathbf{v}_n^{*}) \cdot \mathbf{1}^T (\mathbf{v}_{\langle m-q \rangle _N}^{*} \odot \mathbf{v}_m)
    = \delta_{q,0} , \label{term3}   
\end{align}
Combining \eqref{term1} - \eqref{term3} yields \eqref{average sidelobe of AF}.
\end{proof}
\end{proposition}

\begin{corollary}
    The average squared zero-Doppler slice ($k \ne 0$, $q=0$) of the DP-AF is
    \begin{equation}\label{the average delay cut}
    \setlength{\abovedisplayskip}{4pt}
    \setlength{\belowdisplayskip}{4pt}
     \mathbb{E}(\left| \mathcal{X}(k,0) \right|^2) = N + (\mu_4 - 2) \sum_{i=1}^N \left| \mathbf{u}_i^H  \mathbf{J}_k \mathbf{u}_i\right|^2.
    \end{equation}
    \begin{proof}
        Taking $k \ne 0$ and $q=0$ into \eqref{average sidelobe of AF}, it immediately yields the zero-Doppler slice
\begin{equation}\label{the average delay cut -2}
\begin{aligned}
    &\mathbb{E}(\left| \mathcal{X}(k,0) \right|^2)= N +  \\
    &(\mu_4 - 2)  \sum_{n=1}^N \sum_{m=1}^N \left[ \mathbf{1}^T(\mathbf{v}_{n} \odot \mathbf{v}_n^{*} \odot \mathbf{v}_{m}^{*} \odot \mathbf{v}_m)  e^{-j\frac{2\pi}{N}k(n-m)} \right],
\end{aligned}
\end{equation}
where 
\begin{equation}
    \setlength{\abovedisplayskip}{4pt}
    \setlength{\belowdisplayskip}{4pt}
\begin{aligned} \label{the presentation of vi}
        \mathbf{v}_i = \mathbf{U}^H \mathbf{f}_i^* = 
        \left[ \begin{matrix}
             \mathbf{u}_1^H \mathbf{f}_i^* ,
            \mathbf{u}_2^H  \mathbf{f}_i^* ,
            .
            .
            .,
            \mathbf{u}_N^H \mathbf{f}_i^*     
        \end{matrix}
        \right]^T.
\end{aligned}    
\end{equation}
Plugging \eqref{the presentation of vi} into \eqref{the average delay cut -2}, we can get
\begin{align}\label{the summation of vn and vm -2}
        &\nonumber \sum_{n=1}^N \sum_{m=1}^N \left[ \mathbf{1}^T (\mathbf{v}_{n} \odot \mathbf{v}_n^{*} \odot \mathbf{v}_{m}^{*} \odot \mathbf{v}_m)  e^{-j\frac{2\pi}{N}k(n-m)} \right]\\
        &= \sum_{i=1}^N \left| \mathbf{u}_i^H  \mathbf{J}_k \mathbf{u}_i\right|^2.
\end{align}
Plugging \eqref{the summation of vn and vm -2} into \eqref{the average delay cut -2} results in \eqref{the average delay cut}.
    \end{proof}
\end{corollary}

\begin{corollary}
    The average squared zero-delay slice ($k = 0$, $q\ne0$) of the DP-AF is
    \begin{equation}\label{the average doppler cut}
        \setlength{\abovedisplayskip}{4pt}
        \setlength{\belowdisplayskip}{4pt}
        \mathbb{E}(\left| \mathcal{X}(0,q) \right|^2) =
    N + (\mu_4 - 2)N  \sum_{i=1}^N \left|  \mathbf{u}_i^H \mathrm{Diag}(\mathbf{f}_{q+1}) \mathbf{u}_i \right|^2.
    \end{equation}
    \begin{proof}
Plugging $k = 0$ and  $q \ne 0$ into \eqref{average sidelobe of AF}, the average squared zero-delay slice may be expressed
\begin{equation}\label{the average doppler cut -2}
\begin{aligned}
    &\mathbb{E}(\left| \mathcal{X}(0,q) \right|^2)= N + \\
    &(\mu_4 - 2)  \sum_{n=1}^N \sum_{m=1}^N \left[\mathbf{1}^T (\mathbf{v}_{\langle n-q\rangle _N} \odot \mathbf{v}_n^{*} \odot \mathbf{v}_{\langle m-q \rangle _N}^{*} \odot \mathbf{v}_m) \right].
\end{aligned}
\end{equation}
According to \eqref{the presentation of vi}, we have
\begin{align}\label{the summation of vn_q and vm_q}
    &\nonumber \mathbf{1}^T (\mathbf{v}_{\langle n-q\rangle _N} \odot \mathbf{v}_n^{*} \odot \mathbf{v}_{\langle m-q \rangle _N}^{*} \odot \mathbf{v}_m) \\
    &= \sum_{i=1}^N \mathbf{u}_i^H \mathbf{f}_{\langle n-q \rangle _N}^*  \mathbf{f}_n^T \mathbf{u}_i (\mathbf{u}_i^H \mathbf{f}_{\langle m-q \rangle _N}^*  \mathbf{f}_m^T \mathbf{u}_i)^*.    
\end{align}
Plugging \eqref{the summation of vn_q and vm_q} into \eqref{the average doppler cut -2}, we can get
\begin{align}
    &\nonumber \sum_{n=1}^N \sum_{m=1}^N \mathbf{1}^T (\mathbf{v}_{\langle n-q \rangle _N} \odot \mathbf{v}_n^{*} \odot \mathbf{v}_{\langle m-q \rangle _N}^{*} \odot \mathbf{v}_m) \\
    &= \sum_{i=1}^N \left|  \mathbf{u}_i^H \left(\sum_{n=1}^N \mathbf{f}_{\langle n-q \rangle _N}^*  \mathbf{f}_n^T \right) \mathbf{u}_i \right|^2 ,
\end{align}
where
\begin{equation}
    \setlength{\abovedisplayskip}{4pt}
    \setlength{\belowdisplayskip}{4pt}
    \sum_{n=1}^N \mathbf{f}_{\langle n-q \rangle_N}^*  \mathbf{f}_n^T = \sqrt{N} \mathrm{Diag}(\mathbf{f}_{q+1}).
\end{equation}
This leads to the average squared zero-delay slice \eqref{the average doppler cut}. 
    \end{proof}
\end{corollary}

\begin{corollary}
    When $k\ne 0$ and $q \ne 0$, the average sidelobe level of the DP-AF is
    \begin{equation}\label{the average sidelobe}
    \setlength{\abovedisplayskip}{4pt}
    \setlength{\belowdisplayskip}{4pt}
        \mathbb{E}(\left| \mathcal{X}(k,q) \right|^2)  = N + (\mu_4 - 2)N  \sum_{i=1}^N \left| \mathbf{u}_i^H \mathrm{Diag}(\mathbf{f}_{q+1}) \mathbf{J}_k \mathbf{u}_i \right|^2 .
    \end{equation}
    \begin{proof}
In the case of $k \ne 0$ and $q \ne 0$, \eqref{average sidelobe of AF} may be recast as
\begin{equation}
\begin{aligned}
    &\mathbb{E}(\left| \mathcal{X}(k,q) \right|^2)  = N + (\mu_4 - 2) \times\\
    & \sum_{n=1}^N \sum_{m=1}^N \left[ \mathbf{1}^T(\mathbf{v}_{\langle n-q \rangle _N} \odot \mathbf{v}_n^{*} \odot \mathbf{v}_{\langle m-q \rangle _N}^{*} \odot \mathbf{v}_m) \right]  e^{-j\frac{2\pi}{N}k(n-m)}.
\end{aligned}
\end{equation}
By relying on \eqref{the summation of vn_q and vm_q}, we arrive at
\begin{align}
    &\nonumber \sum_{n=1}^N \sum_{m=1}^N \left[ \mathbf{1}^T (\mathbf{v}_{\langle n-q \rangle _N} \odot \mathbf{v}_n^{*} \odot \mathbf{v}_{\langle m-q \rangle _N}^{*} \odot \mathbf{v}_m) \right]  e^{-j\frac{2\pi}{N}k(n-m)} \\
    &= \sum_{i=1}^N \left|  \mathbf{u}_i^H \left(\sum_{n=1}^N \mathbf{f}_{\langle n-q \rangle _N}^*  \mathbf{f}_n^T  e^{-j\frac{2\pi}{N}k(n-1)} \right) \mathbf{u}_i\right|^2.
\end{align}
where
\begin{equation}
    \setlength{\abovedisplayskip}{4pt}
    \setlength{\belowdisplayskip}{4pt}
    \sum_{n=1}^N \mathbf{f}_{\langle n-q \rangle _N}^*  \mathbf{f}_n^T  e^{-j\frac{2\pi}{N}k(n-1)} = \sqrt{N} \mathrm{Diag}(\mathbf{f}_{q+1}) \mathbf{J}_k,
\end{equation}
leading to the average sidelobe level of the DP-AF in \eqref{the average sidelobe}. 
    \end{proof}
\end{corollary}

According to the Cauchy-Schwarz inequality, we have
\begin{align}
        &\nonumber\left| \mathbf{u}_i^H  \mathbf{J}_k \mathbf{u}_i\right|^2 
         \leq   1, \quad
        \left|  \mathbf{u}_i^H \mathrm{Diag}(\mathbf{f}_{q+1}) \mathbf{u}_i \right|^2 
        \leq  \frac{1}{N}, \\
        &\left| \mathbf{u}_i^H \mathrm{Diag}(\mathbf{f}_{q+1}) \mathbf{J}_k \mathbf{u}_i \right|^2 
    \leq  \frac{1}{N}.\label{the cauchy}
\end{align}
Combining \eqref{the average delay cut}, \eqref{the average doppler cut}, \eqref{the average sidelobe}, and \eqref{the cauchy}, it becomes apparent that for QAM and PSK constellations whose kurtosis values are less than 2, the corresponding average sidelobe level of their DP-AFs is between $(\mu_4 - 1)N$ and $N$.

\subsection{EISL}

\begin{thm} [Invariance of EISL] For all constellations and modulation schemes, the normalized EISL is a constant value $N-1$.
\begin{equation}
    \setlength{\abovedisplayskip}{5pt}
    \setlength{\belowdisplayskip}{5pt}
    \frac{\sum_{k=0}^{N-1} \sum_{q=0}^{N-1}\mathbb{E}(\left| \mathcal{X}(k,q) \right|^2) - \mathbb{E}(\left| \mathcal{X}(0,0) \right|^2)}{\mathbb{E}(\left| \mathcal{X}(0,0) \right|^2)} = N -1.
\end{equation}
\end{thm}
\begin{proof}

According to \eqref{average sidelobe of AF}, it's clear that the summation of EISL and average mainlobe can be divided into three terms.

The first and third summation terms respectively are
\begin{equation} \label{Part1}
    \setlength{\abovedisplayskip}{4pt}
    \setlength{\belowdisplayskip}{4pt}
    \sum_{k=0}^{N-1} \sum_{q=0}^{N-1}  \sum_{n=1}^N \sum_{m=1}^N \delta_{m,n} e^{-j \frac{2\pi}{N} k (n-m)} = N^3.
\end{equation}
\begin{equation}\label{Part3}
    \setlength{\abovedisplayskip}{4pt}
    \setlength{\belowdisplayskip}{4pt}
    \sum_{k=0}^{N-1} \sum_{q=0}^{N-1}  \sum_{n=1}^N \sum_{m=1}^N \delta_{q} e^{-j \frac{2\pi}{N} k (n-m)}
    = N^2.
\end{equation}
The second summation term is more complex. We first consider the summations of $m$ and $n$, as shown in \eqref{the summation of m and n of part2}. 
    \begin{align}\label{the summation of m and n of part2}
        &\nonumber \sum_{n=1}^N \sum_{m=1}^N \mathbf{1}^T(\mathbf{v}_{\langle n-q\rangle_N} \odot \mathbf{v}_n^{*} \odot \mathbf{v}_{\langle m-q\rangle_N}^{*} \odot \mathbf{v}_m)  e^{-j \frac{2\pi}{N} k (n-m)} \\
        & = N   \sum_{i=1}^N \left| \mathbf{f}_{k+1}^H \mathrm{Diag}(\mathbf{F}_N \mathbf{u}_i)^H \mathbf{J}_{q} \mathbf{F}_N \mathbf{u}_i \right|^2.
    \end{align}  
Using \eqref{the summation of m and n of part2}, the second summation term is shown in \eqref{Part2}.
\begin{align}\label{Part2}
        &\nonumber \sum_{k=0}^{N-1} \sum_{q=0}^{N-1} \sum_{n=1}^N \sum_{m=1}^N (\mu_4 -2) \mathbf{1}^T \\
        &(\mathbf{v}_{\langle n-q\rangle_N} \odot \mathbf{v}_n^{*} \odot \mathbf{v}_{\langle m-q\rangle_N}^{*} \odot \mathbf{v}_m) e^{-j \frac{2\pi}{N} k (n-m)} = (\mu_4 - 2)N^2.
\end{align}
Combining \eqref{Part1}, \eqref{Part3}, and \eqref{Part2}, we arrive at
\begin{equation}
    \setlength{\abovedisplayskip}{4pt}
    \setlength{\belowdisplayskip}{4pt}
    \sum_{k=0}^{N-1} \sum_{q=0}^{N-1}  \mathbb{E}(\left| \mathcal{X}(k,q) \right|^2) = N^3 + (\mu_4 -1)N^2.
\end{equation}
This completes the proof.
\end{proof}

According to Theorem 1, when the EISL of the DF-AF is used as the sensing metric, communication-centric ISAC signals formulated by any constellation and modulation basis yield identical ambiguity performance. Hence, no optimal sensing waveform exists in the sense of minimizing the EISL.

\section{Simulations}
In this section, we further validate the theoretical analysis and key conclusions through numerical simulations. All results are obtained by averaging over 1000 independent random realizations.

We begin by statistically characterizing the AF of random communication signals. In Fig.~\ref{fig: the 2D AF performances compare}, the average sidelobe levels of 16-QAM symbols under different communication waveforms are compared. It can be observed that the simulation results closely match the theoretical predictions. The consistency between simulation and theory confirms the validity of the analytical derivations. 

Then, the constant behavior of the normalized EISL is validated in Fig. \ref{fig:Traditional AF EISL Compare different N 16QAM and 16PSK}. As shown, the sum of the normalized EISL and normalized mainlobe levels remains equal to $N$ for all five considered waveform types under both 16-QAM and 16-PSK constellations. These results empirically verify the theoretical claim made in Theorem 1.
\begin{figure}[ht]
    \centering
    \vspace{-0.3cm}
        \subfigure[OFDM]{
        \includegraphics[width = 0.45\columnwidth]{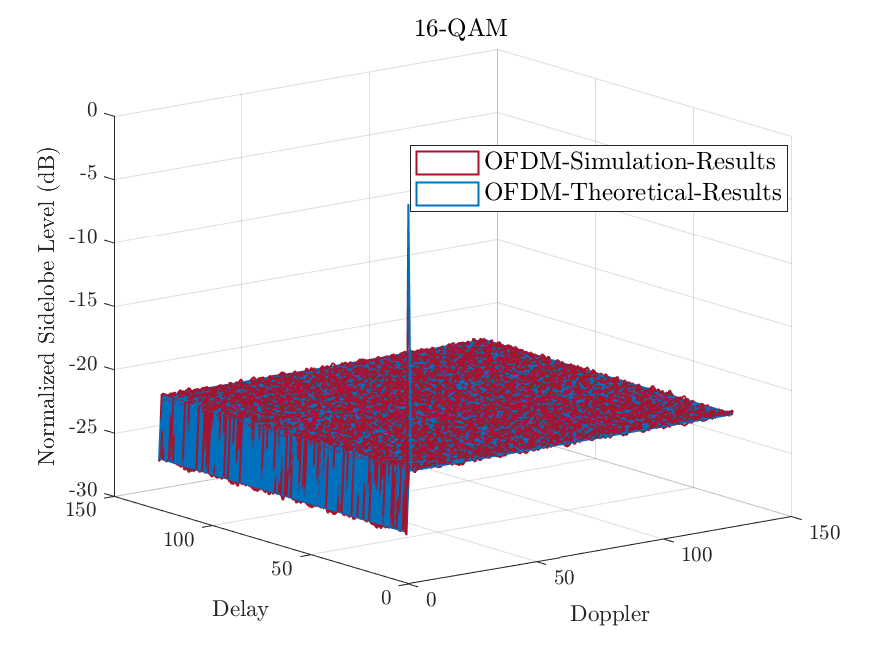} 
        \label{fig: AF of 16QAM under OFDM compare}    
    }
    \subfigure[SC]{
        \includegraphics[width=0.45\columnwidth]{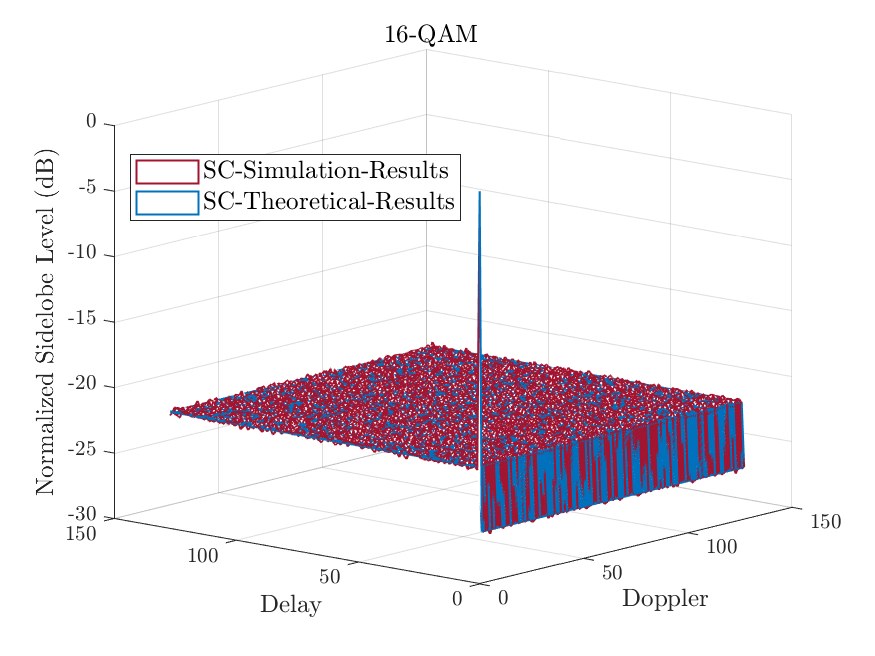}
        \label{fig: AF of 16QAM under SC compare}    
    }
    \vspace{-0.3cm}
    \caption{The DP-AF performances of 16-QAM  under OFDM and SC signals with $N = 128$.}
    \vspace{-0.3cm}
    \label{fig: the 2D AF performances compare}
\end{figure}

\begin{figure}[ht]
    \centering
    \vspace{-0.3cm}
    \subfigure[16-QAM]{
    \includegraphics[width=0.45\columnwidth]{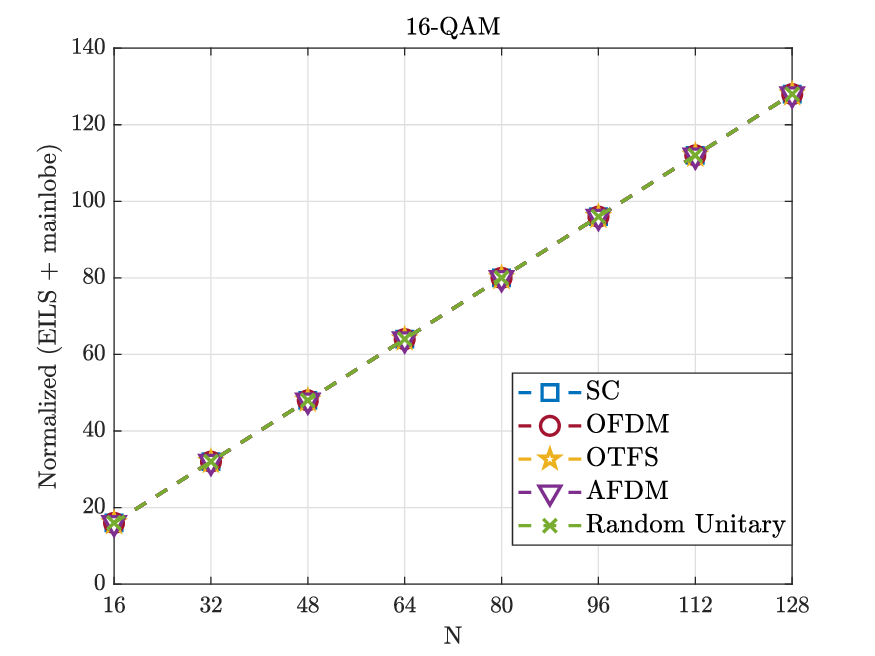}
    \label{fig:Traditional AF EISL Compare different N 16QAM}}
    \subfigure[16-PSK]{
    \includegraphics[width=0.45\columnwidth]{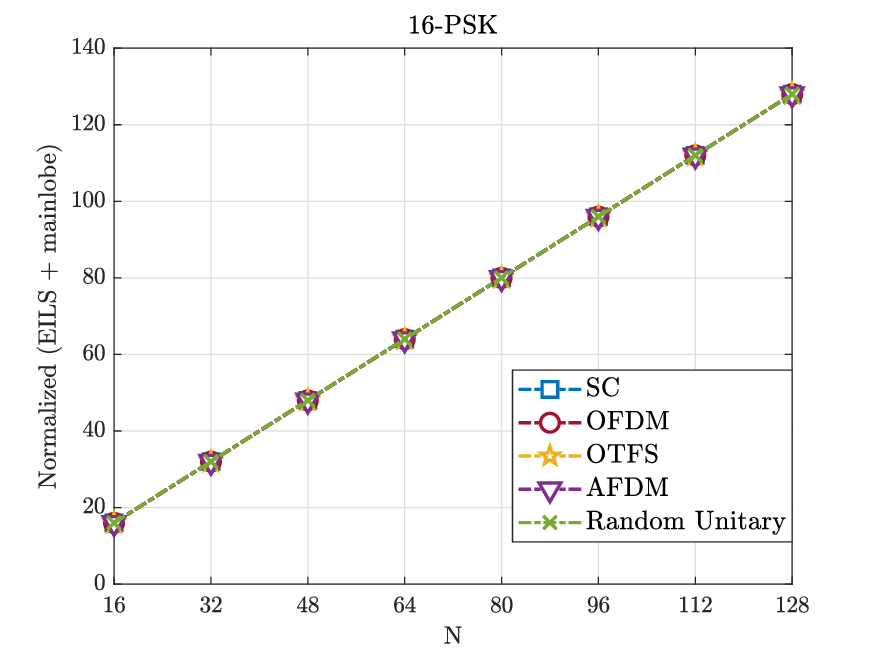}
    \label{fig:Traditional AF EISL Compare different N 16PSK}
    }
    \vspace{-0.3cm}
    \caption{The normalized EISL performances of 16-QAM and 16-PSK compared under different communication signals with varying $N$.}
    \label{fig:Traditional AF EISL Compare different N 16QAM and 16PSK}
\end{figure}

\section{Conclusions}
This paper analyzed the DP-AF of random communication-centric ISAC waveforms. We derived closed-form expressions for the average sidelobe level under arbitrary constellations and proved that the normalized EISL of the DP-AF is invariant across all constellations and modulation schemes.

\vfill\pagebreak


\bibliographystyle{IEEEbib}
\balance
\bibliography{refs}

\end{document}